\documentclass[letterpaper, 10 pt, conference]{ieeeconf}  % Comment this line out if you need a4paper

\IEEEoverridecommandlockouts                              % This command is only needed if 
\usepackage{comment}
\usepackage{graphics} % for pdf, bitmapped graphics files
\usepackage{epsfig} % for postscript graphics files
\usepackage{times} % assumes new font selection scheme installed
\usepackage{amsmath} % assumes amsmath package installed
\usepackage{amssymb}  % assumes amsmath package installed
\usepackage{amsfonts}
\usepackage{color}
\usepackage{algorithm,algpseudocode}
\usepackage{soul}
\usepackage{tikz}
\usepackage{tikz-network}
\usepackage{url}
\usepackage{csquotes}
\usetikzlibrary{backgrounds}
\tikzset{background rectangle/.style={fill=yellow!30}}
\usetikzlibrary{arrows.meta, positioning}

\let\labelindent\relax
\usepackage{enumitem}  
\setlist[description]{leftmargin=\parindent}
\usepackage[mathscr]{eucal}
\usepackage{subfigure}

\usepackage{theorem}
\newtheorem{theorem}{Theorem}
\newtheorem{definition}{Definition}
\newtheorem{proposition}{Proposition}
\newtheorem{lemma}{Lemma}

\newtheorem{assumption}{Assumption}
\newtheorem{problem}{Problem}

{\theorembodyfont{\upshape}
	\newtheorem{remark}{Remark}
	\newtheorem{example}{Example}
	
}

\def\qed{\hfill \vrule height 5pt width 5pt depth 0pt \medskip}

\def\im{{\rm Im}}

\newcommand{\Am}{\mathcal{A}}
\newcommand{\Bm}{\mathcal{B}}
\newcommand{\Cm}{\mathcal{C}}
\newcommand{\Dm}{\mathcal{D}}
\newcommand{\Em}{\mathcal{E}}
\newcommand{\Gm}{\mathcal{G}}
\newcommand{\Hm}{\mathcal{H}}

\newcommand{\Jm}{\mathcal{J}}

\newcommand{\Pm}{\mathcal{P}}

\newcommand{\Sm}{\mathcal{S}}
\newcommand{\Tm}{\mathcal{T}}
\newcommand{\Vm}{\mathcal{V}}
\newcommand{\Wm}{\mathcal{W}}

\newcommand{\Zm}{\mathcal{Z}}

\newcommand{\beq}{\begin{equation}}
	\newcommand{\eeq}{\end{equation}}
\newcommand{\beqa}{\begin{eqnarray}}
	\newcommand{\eeqa}{\end{eqnarray}}
\newcommand{\beqan}{\begin{eqnarray*}}
	\newcommand{\eeqan}{\end{eqnarray*}}

\newcommand{\bite}{\begin{itemize}}
	\newcommand{\eite}{\end{itemize}}
\newcommand{\benu}{\begin{enumerate}}
	\newcommand{\eenu}{\end{enumerate}}

\definecolor{darkpastelgreen}{rgb}{0.01, 0.75, 0.24}

\usepackage{placeins}
\usepackage{bm}
\usepackage{lipsum}
\algnewcommand{\LeftComment}[1]{\Statex \(\triangleright\) #1}

\title{\LARGE \bf
	A Polynomial-Time Algorithm for the Minimal Input/Output Node Cardinality Disturbance Decoupling Problem over Networks via Output Feedback
}

\author{
	\thanks{Work supported in part by the Swedish Research Council (grant n. 2024-04772 to C.A.) and by the ELLIIT framework program at Link\"oping University. }%
	Luca Claude Gino Lebon, Jorke de Vlas, Johannes K. Fichte, and Claudio Altafini% <-this % stops a space
	\thanks{L. C. G. Lebon and C. Altafini are with the Department of Electrical Engineering, Division of Automatic Control,
		Linköping University, SE-581 83 Linköping, Sweden
		{\tt\small luca.lebon@liu.se, claudio.altafini@liu.se}. Jorke de Vlas is with the Department of Computer and Information Science, Division of Artificial Intelligence and Integrated Computer Systems, Linköping University, SE-581 83 Linköping, Sweden {\tt\small jorke.de.vlas@liu.se}. Johannes K. Fichte is with CRIL (UMR 8188), Artois University and CNRS, F-62300 Lens, France {\tt\small johannes.fichte@cril.fr}.}}

\begin{document}

	\maketitle
	\thispagestyle{empty}
	\pagestyle{empty}
	
	%%%%%%%%%%%%%%%%%%%%%%%%%%%%%%%%%%%%%%%%%%%%%%%%%%%%%%%%%%%%%%%%%%%%%%%%%%%%%%%%
	\begin{abstract}
		For networked systems, we propose a polynomial-time algorithm for solving the minimal input/output node cardinality disturbance decoupling problem via output feedback. The algorithm computes an optimal solution, namely, a pair of minimum-cardinality sets of control input and measurement output nodes that achieve disturbance decoupling, by solving a min-cut/max-flow problem on an extended network.
	\end{abstract}

	%%%%%%%%%%%%%%%%%%%%%%%%%%%%%%%%%%%%%%%%%%%%%%%%%%%%%%%%%%%%%%%%%%%%%%%%%%%%%%%%
	
	\section{Introduction}
	In this paper, we investigate the Disturbance Decoupling Problem via Output Feedback (DDPOF) for networked systems, where disturbances act on a subset of nodes and their effects must be isolated from a designated set of target nodes. For linear systems, disturbance decoupling is classically studied within geometric control theory \cite{basile1992controlled,wonham2013linear,trentelman2012control}. When the full state is not directly available for feedback, the problem can be addressed by designing an output feedback law of the form \(u=-Gy\), where \(y=Cx\) is a suitably selected output vector. Necessary and sufficient conditions for the solvability of the DDPOF can be expressed in terms of the existence of a subspace that contains the image of the disturbance matrix and is contained in the kernel of the target matrix, while being both controlled invariant and conditioned invariant \cite{hamano1975localization}.
	
	In \cite{lebon2026geometric}, the DDP is reformulated for networked systems under the assumption that disturbances, targets, control inputs, and output measurements are associated with individual nodes, rather than arbitrary linear combinations of states. Under this node-based setting, graphical sufficient conditions for solvability are established in terms of properties of the network nodes. In particular, if the control input and measurement output nodes contain, respectively, the out-boundary and the in-boundary of a node set \(\Wm\), then \(\Wm\) is both a controlled and a conditioned invariant set of nodes. Consequently, the DDPOF is solvable whenever \(\Wm\) contains all disturbance nodes and is disjoint from the target nodes. Building on this characterization, \cite{lebon2026geometric} formulates the minimal input/output node cardinality DDPOF, which seeks to minimize the number of control input and measurement output nodes required to achieve disturbance decoupling. The optimal solutions are mathematically characterized, and a corresponding output feedback gain, known as a “friend” of \(\Wm\), is derived in closed form. The framework is also applied to disturbance decoupling in Power Networks in \cite{IFAC_26}.
	
	Despite these results, while polynomial-time algorithms have been developed for the corresponding minimal node cardinality state-feedback and dynamic-feedback problems \cite{lebon2026geometric}, an efficient algorithm for solving the minimal input/output node cardinality DDPOF has, to the best of our knowledge, not yet appeared in the literature. This paper fills this algorithmic gap by introducing an extended network that encodes the graphical constraints imposed by the optimality characterization of the minimal input/output node cardinality DDPOF, and then applying a min-cut/max-flow algorithm on the extended network in order to identify an optimal solution. Related applications of flow-based methods to control problems include functional target controllability \cite{Commault2020} and the computation of the infinite structure of structured systems \cite{hovelaque1996analysis}. We prove that the proposed algorithm always terminates in polynomial time and returns an optimal solution.
	
	A MATLAB implementation is available in \cite{lebon_2026_23476722}, complementing the toolbox accompanying~\cite{lebon2026geometric}.

	\section{Background material}
	The material in this section is taken from standard textbooks on geometric control; see, e.g., \cite{basile1992controlled,trentelman2012control,wonham2013linear} and adapted for systems over networks in \cite{lebon2026geometric}.
	
	Consider a digraph $ \Gm = ( \Vm, \, \Em , A )$, where $ \Vm =\{ v_1, \, \ldots, v_n\}$ is a set of $n$ nodes and $\Em $ a set of $q$ edges over $ \Vm$, with the convention that $ (v_i , \, v_j) \in \Em $ means $ v_i \to v_j $, i.e., $ v_i $ is the tail and $ v_j $ is the head of the edge. 
	$ A\in\mathbb{R}^{n\times n}$ is the adjacency matrix associated with $ \Gm$: $ A_{ji} $ is the weight of the edge $ (v_i , \, v_j) \in \Em $. 
	Denote by $ \Am $ the class of adjacency matrices of $ \Gm$ having ``support" in $ \Em$: $\Am = \{ A \in\mathbb{R}^{n\times n} \; \text{ s.t.} \; A_{ji}\neq 0\;  \text{iff} \; (v_i , \, v_j) \in \Em\; \text{and} \; A_{ji}= 0\;  \text{iff} \; (v_i , \, v_j) \notin \Em \}$. Given a graph $ \Gm$, a subset $ \Zm \subseteq \Vm$ is called a terminal subset of nodes if there is no outgoing edge from $ \Zm$, i.e., for all edges $ (v_i , \, v_j) \in \Em $ if $ v_i \in \Zm $ then also $ v_j \in \Zm$.
	A directed path $ \Pm$ from node $v_i$ to node $v_j$ is a sequence of distinct nodes $\{v_i, v_{i+1}, \dots , v_{j-1}, v_j\}$ s.t. $(v_k,v_{k+1})\in \mathcal{E}$ $\forall\,k\in\{i,i+1, \dots, j-1\}$: $\Pm = \{ v_i,  \dots , v_j \} $ (or $ \Pm^{v_i, v_j} $ when we need to specify the extremities of the path). 
	The number of edges in the path is called the length of the path. A cut $(S,T)$ of $\Gm$ is a partition of the node set $\Vm$ into two disjoint subsets $S$ and $T$, with $S\cup T=\Vm$. The (directed, omitted hereafter) edge cut (also called (directed) cut-set) induced by $(S,T)$ is defined as $\delta(S,T):=\{(u,v)\in\Em:u\in S,\ v\in T\}$, i.e., the set of edges directed from $S$ to $T$. 
	The value of a cut is the sum of the weights of the edges in $\delta(S,T)$. If all edges have unit weight, the value of $(S,T)$ equals the cardinality of $\delta(S,T)$. 
	Given two nodes $ s$ (source) and $ t$ (sink), an $s$-$t$ cut is an edge cut $(S,T)$ such that $s\in S$ and $t\in T$. A minimum $s$-$t$ cut is an $s$-$t$ cut whose value, called the min-cut value, is no greater than that of any other $s$-$t$ cut. A set $X\subseteq\Vm$ is an $s$-$t$ node cut (also called $s$-$t$ vertex cut) if removing the nodes in $X$, together with their incident edges, eliminates every directed path from $s$ to $t$. More generally, for disjoint sets $S,T\subseteq\Vm$, a set $X\subseteq\Vm$ is a node cut (equiv., vertex cut) separating $S$ and $T$ if removing the nodes in $X$, together with their incident edges, eliminates every directed path from any node in $S$ to any node in $T$. 
	
	In the following, we explicitly fix a basis of canonical vectors for a vector subspace associated with a set of nodes.
	Let $ e_i $ be an elementary column vector (i.e., the vector having $1$ on the $i$-th slot and $0$ elsewhere).
	\begin{definition}[\textbf{Subspace associated with a set of nodes}]\label{def:sub}
		Given a set of nodes $\Zm\subseteq\Vm$, the vector subspace of $\mathbb{R}^n$ associated with $\Zm$ is $\textnormal{span}\{e_i\;\text{s.t.}\;v_i\in\Zm\}$.
	\end{definition}
	Throughout the paper, we identify $ \Zm $ and the associated vector subspace, and sometimes write $ x \in \Zm $. For any $ \Zm \subset \Vm$, the notation $ \Zm^\perp$ is used for both the complement set $ \Vm \setminus \Zm$ and the subspace orthogonal to (the vector space) $ \Zm $ in $ \mathbb{R}^n$. For each subspace $ \Zm$, $ \exists $ a full column rank matrix $ Z$ s.t. $ \im Z = \Zm $ and $ \ker Z =\{0\}$. Such matrix is called a basis matrix of $\Zm$. When $ \Zm $ is a set of nodes, we can take $ Z$ as a collection of elementary column vectors (i.e., one column $e_i$ for each node $ v_i \in \Zm$).
	As a consequence, the sum of vector subspaces and their union can be identified: $ \Zm_1 + \Zm_2 \sim \Zm_1 \cup \Zm_2 $.
	%Furthermore, if $ \Sm \subset \Vm$ is another set, then $ \Sm \setminus \ker (Z)^\perp = \Sm \cap \ker (Z)$. 
	On $ \Gm = ( \Vm, \, \Em , A )$, for any subset of nodes $ \Wm\subseteq\Vm$ we can define the in- and out-boundaries as follows.
	\begin{definition}\label{def:boundary}
		The in-boundary of $\Wm\subseteq\Vm$ in the digraph $ \Gm = ( \Vm, \, \Em , A )$ is the set of nodes in $\Wm$ that have at least one outgoing edge leading to a node in $\Wm^\perp$, i.e., $\partial_-(\Wm,A):=\{v_i\in\Wm\;\text{s.t.}\;\exists\,(v_i,v_j)\in\Em,\, v_j\in\Wm^\perp\}$. Similarly, the out-boundary of $\Wm$ is the set of nodes in $\Wm^\perp$ that have at least one incoming edge from $\Wm$, i.e., $\partial_+(\Wm,A):=\{v_j\in\Wm^\perp\;\text{s.t.}\;\exists\,(v_i,v_j)\in\Em,\, v_i\in\Wm\}$.
	\end{definition}
	From the preceding definition, it directly follows that $\partial_\pm(\Wm,A)=\partial_\mp(\Wm^\perp,A^\top)$.

	\subsection{Linear systems over networks}
	
	In $ \Gm $, we can assume that two subsets of nodes play the role of inputs and outputs, denoted, respectively, $ \Bm \subseteq \Vm$ and $ \Cm \subseteq \Vm$.
	Then we can form a linear time-invariant control system having a state update matrix equal to $A$, an input matrix $B\in\mathbb{R}^{n\times m}$ composed of $m$ elementary column vectors corresponding to $ \Bm=\im B $, and an output matrix $C\in\mathbb{R}^{p\times n}$ having $p$ elementary rows corresponding to $ \Cm=\im C^\top=(\ker C)^\perp$:
	\beq
	\begin{split}
		\dot x & =  Ax + B u, \\ 
		y & = C x ,
	\end{split}
	\label{eq:lin-syst1}
	\eeq
	where we have associated one state variable $ x_i \in \mathbb{R} $ of the vector $ x = [ x_1 \ldots x_n]^\top $ to each node $ v_i $ of the graph.
	\subsection{Invariance, controlled invariance and conditioned invariance for node sets}
	
	In this section we formulate the concepts of invariance for sets of nodes, rather than for subspaces, as it is standard in the literature (see, e.g., \cite{basile1992controlled}).
	
	\begin{definition}
		\label{def:A-invar1}
		A subset $ \Zm \subset \Vm$ is an invariant (or $A$-invariant) set of nodes for~\eqref{eq:lin-syst1} if $  A \Zm \subseteq \Zm $.
	\end{definition}

	Definition~\ref{def:A-invar1}, a standard concept also for subspaces, corresponds to $Ax\in\Zm$, $\forall\,x\in\Zm$. For network systems like~\eqref{eq:lin-syst1} we can have also an equivalent topological characterization, valid for the entire class $ \Am$. 
	%The proof of the next Proposition follows as a special case of Proposition~\ref{prop:AB-invar-equiv-char} below. 
	%\bluenote{LL: lascerei questo commento.}

	\begin{proposition}[\textbf{\cite[Prop.~1]{lebon2026geometric}}]
		\label{prop:invA}
		Consider the system~\eqref{eq:lin-syst1}. A subset $ \Zm \subset \Vm $ is an $A$-invariant set of nodes iff any of the following equivalent conditions is met:
		\benu[label=(\alph*)]
		\item $ \Zm$ is a terminal subset of nodes;
		\item If $ Z $ is the basis matrix of $ \Zm $ (i.e., $ \im Z = \Zm $ and $Z$ is full column rank), $ \exists$ a matrix $X$ s.t. $AZ = ZX$. 
		%\rednote{CA: questa va bene qua? E' corretta?}\,\bluenote{LL: Si. Ho aggiunto che i generatori siano lin.ind.}
		\eenu
	\end{proposition}

	The concept of controlled invariance is used to expand the notion of invariance (see, e.g., Chapter 4 of \cite{trentelman2012control}).

	\begin{definition}
		\label{def:contr-invar}
		A subset $ \Zm \subset \Vm$ is a controlled invariant (or $(A,B)$-invariant) set of nodes for~\eqref{eq:lin-syst1} if $ A \Zm \subseteq \Zm + \Bm $.
	\end{definition}

	The following proposition provides three equivalent characterizations of $ (A, \, B)$-invariance for sets of nodes.
	\begin{proposition}[\textbf{\cite[Prop.~2]{lebon2026geometric}}]
		\label{prop:AB-invar-equiv-char}
		$ \Zm \subset \Vm$ is an $ (A,B)$-invariant set of nodes for~\eqref{eq:lin-syst1} iff any of the following equivalent conditions is met:
		\benu[label=(\alph*)]
		\item $\forall\,(v_i , \, v_j) \in \Em $, $ v_i \in \Zm \; \Longrightarrow \; v_j \in \Zm \cup \Bm$;\label{item:a_prop:AB-invar-equiv-char}
		\item If $ Z $ is the basis matrix of $ \Zm $, $ \exists $ matrices $X$, $U$ s.t. $AZ = Z X +BU$;\label{item:b_prop:AB-invar-equiv-char}
		\item $\exists$ a matrix $F\in\mathbb{R}^{m\times n}$, called a ``friend'' of $\Zm$, s.t. $\Zm$ is $(A-BF)$-invariant, i.e., $(A-BF)\Zm\subseteq\Zm$.\label{item:c_prop:AB-invar-equiv-char}
		
		%$ \exists $ matrices $X\in\mathbb{R}^{r\times r}$, $U\in\mathbb{R}^{m\times r}$ such that $AZ = Z X +BU$.
		\eenu
	\end{proposition}
	
	\begin{definition}
		A subset $ \Sm \subset \Vm$ is a conditioned invariant (or $ (C,A)$-invariant) set of nodes for~\eqref{eq:lin-syst1} if $ A ( \Sm \setminus\Cm)\subseteq \Sm$ (with $\Cm=(\ker C)^\perp$ and $C$ with elementary rows).
	\end{definition}
	Similarly to above, equivalent characterizations of $ (C,A)$-invariance can be formulated.
	\begin{proposition}[\textbf{\cite[Prop.~3]{lebon2026geometric}}]\label{prop:CA-invar-equiv-char}
		$\Sm $ is a $ (C,A)$-invariant set of nodes if and only if any of the following equivalent conditions is met:
		\benu[label=(\alph*)]
		\item $\forall\,(v_i , \, v_j) \in \Em $, $ v_i \in \Sm$, $ v_i \notin \Cm$ $ \Longrightarrow \; v_j \in \Sm$;\label{item:a_prop:CA-invar-equiv-char}
		\item $\exists$ a matrix $H\in\mathbb{R}^{n\times p}$, called a ``friend'' of $\Sm$, s.t. $\Sm$ is $(A-HC)$-invariant, i.e., $(A-HC)\Sm\subseteq\Sm$.\label{item:b_prop:CA-invar-equiv-char}
		\eenu
	\end{proposition}

	%\subsubsection[Closed-loop system and $(C,A,B)$-pairs]{Closed-loop system and $\bm{(C,A,B)}$-pairs}

	\begin{proposition}\label{prop:cont_cond}
		A set of nodes $\Wm$ is $(A,B)$-invariant and $(C,A)$-invariant iff it is $A$-invariant or each edge exiting $\Wm$ is of the type $(v_i,v_j)$ with $v_i\in\Cm$ and $v_j\in\Bm$.
	\end{proposition}
	\begin{proof}
		One direction is proven directly from the definitions of $(A,B)$-invariant and $(C,A)$-invariant node set. In fact, a terminal set of nodes $\Wm$ is $A$-invariant, therefore also $(A,B)$-invariant and $(C,A)$-invariant. Assume that a set of nodes $\Wm$ is not $A$-invariant and define the set of all the edges exiting $\Wm$ as $\overline{\Cm}\times\overline{\Bm}\subseteq\Em$, $\overline{\Cm}\subseteq\Wm$, $\overline{\Bm}\subseteq\Wm^\perp$. If $v_i\in\Cm$, $\forall\,v_i\in\overline{\Cm}$ then $\Wm$ is a $(C,A)$-invariant node set, whereas if $v_j\in\Bm$, $\forall\,v_j\in\overline{\Bm}$ then $\Wm$ is an $(A,B)$-invariant node set.
		
		For the other direction, assume by contradiction $\exists$ an edge exiting $\Wm$ not of the type $(v_i,v_j)$ with $v_i\in\Cm$ and $v_j\in\Bm$. Then either condition~$\ref{item:a_prop:AB-invar-equiv-char}$ of Proposition~\ref{prop:AB-invar-equiv-char}, or condition~$\ref{item:a_prop:CA-invar-equiv-char}$ of Proposition~\ref{prop:CA-invar-equiv-char}, or both, are violated.
	\end{proof}

	\subsection{Standard disturbance decoupling problem via output feedback over networks}
	\label{sec:standard-DDP}

	For a given network system \eqref{eq:lin-syst1}, we can consider two additional sets of nodes $ \Dm, \Tm \subset \mathcal{V}$, of cardinality $ | \Dm|$ and $ | \Tm | $, corresponding, resp., to nodes affected by disturbances and target nodes. Each disturbance $ w_i $, $ i\in\{1, \ldots, |\Dm| \}$, is assumed to act independently on a node of $ \Dm$.
	Both $ \Dm $ and $ \Tm $ can be identified with elementary matrices $ D \in \mathbb{R}^{n\times | \Dm|} $ and $ T\in \mathbb{R}^{| \Tm|\times n} $, i.e., matrices composed of elementary columns, resp. rows, representing subselections of nodes in $ \Vm$. 
	We can then associate to $ \Dm $ and $ \Tm $ vector subspaces, resp. $ \im D =\Dm $ and $ \ker T = \Vm\setminus\Tm$. Denoting $ w = [ w_1 \ldots w_{|\Dm|}]^\top $, we obtain the following generalization of \eqref{eq:lin-syst1}:
	\beq
	\begin{split}
		\dot x & =  Ax + B u + D w, \\
		y & = Cx,\\
		z & = Tx .
	\end{split}
	\label{eq:lin-syst2}
	\eeq
	For the system~\eqref{eq:lin-syst2}, we consider the following problem.
	\begin{problem}[\textbf{DDPOF}]\label{problem:DDPOF}
		Given the system~\eqref{eq:lin-syst2}, find a control law $u= - Gy $ that renders the target nodes in $\mathcal{T}$ unaffected by the disturbances in $\mathcal{D}$.
	\end{problem}

	\section{Minimal Input/Output Cardinality DDPOF}

	In order to avoid degenerate situations in the problem, we make the following assumption on the sets $\Dm$ and $\Tm$.
	\begin{assumption}\label{ass:2}
		$\Dm\cap\Tm = \emptyset $ (equiv. $TD=0$).
	\end{assumption}

	The solvability of Problem~\ref{problem:DDPOF} is characterized by the following standard result.
	\begin{theorem}[\textbf{\cite[Thm.~3.2]{hamano1975localization}}]\label{thm:ddp}
		For the system~\eqref{eq:lin-syst2} the DDPOF is solvable iff $\im D\subseteq\Wm^\ast\subseteq\ker T$, where $\Wm^\ast$ is an $(A,B)$- and $(C,A)$-invariant subspace.
	\end{theorem}
	In \cite{lebon2026geometric}, a sufficient condition for the solvability of Problem~\ref{problem:DDPOF} is stated in terms of the corresponding node set $\Wm$.
	
	\begin{theorem}[\textbf{\cite[Thm.~3]{lebon2026geometric}}]\label{thm:DDPN}
		Under Assumption~\ref{ass:2}, the DDPOF is solvable for the system~\eqref{eq:lin-syst2} if any of the following equivalent conditions is met:
		\benu[label=(\alph*)]
		\item $\exists\,\Wm$ $(A,B)$-invariant and $(C,A)$-invariant set of nodes s.t. $\Dm\subseteq\Wm\subseteq\Vm\setminus\Tm$;\label{item:a_thm:ddpof}
		\item  Let $\Pm_1,\dots,\Pm_\ell$ be all possible $\Dm$-to-$\Tm$ paths. For every path $\Pm_j=\{v_1,\dots,v_{\kappa_j}\}$ of length $\kappa_j-1\geq1$, with $v_1\in\Dm$ and $v_{\kappa_j}\in\Tm$, $\exists$ at least one sub-path of length 1 of the form $\{v_{p},v_{p+1}\}$ with $v_{p}\in\Cm$ and $v_{p+1}\in\Bm$, $p\in\{1,\dots,\kappa_j-1\}$, $j\in\{1,\dots,\ell\}$.\label{item:c_thm:ddpof}
		\eenu
	\end{theorem}
	
	\begin{remark}
		The node set $\Wm$ is determined by $\Bm$ and $\Cm$ since it can be constructed as an extension of the minimal $(C,A)$-invariant node set containing $\Dm$, which is uniquely determined by $\Cm$, or as a reduction of the maximal $(A,B)$-invariant node set contained in $\Tm^\perp$, which is uniquely determined by $\Bm$ (see the proof of \cite[Thm.~3]{lebon2026geometric}).
	\end{remark}
	As a consequence of Theorem~\ref{thm:DDPN} we can formulate the following minimal node-cardinality problem.

	\begin{problem}[\textbf{Minimal Input/Output Cardinality DDPOF}]
		\label{problem:DDPOF_N}
		Given a network $\Gm$ and sets $ \Dm$, $ \Tm$, find a minimal cardinality set of input nodes $\Bm$ and output nodes $\Cm$ that solves the DDPOF for the system \eqref{eq:lin-syst2}.
	\end{problem}

	In practice, solving Problem~\ref{problem:DDPOF_N} corresponds to finding matrices $ B \in \mathbb{R}^{n\times m}$ and $C^\top\in \mathbb{R}^{n\times p}$ of elementary columns s.t. $ \im B =\Bm $ and $ \im C^\top =\Cm $.
	\subsection{Solution to the minimal input/output cardinality DDPOF}

	Exploiting Theorem~\ref{thm:DDPN}, Problem~\ref{problem:DDPOF_N} can be reformulated as an optimization problem:
	
	\begin{equation*}
		\textbf{Problem~\ref{problem:DDPOF_N}:} \qquad \Biggl\{
		\begin{aligned}
			& \min_{\Bm,\Cm} \quad |\Bm|+|\Cm|\\
			& \text{subject to:} \quad \Dm\subseteq\Wm\subseteq\Tm^\perp
			,\end{aligned}
	\end{equation*}
	where $\Wm$ is an $(A,B)$- and $(C,A)$-invariant set of nodes.
	
	In Section~IV of \cite{lebon2026geometric}, we characterize the feasible (i.e., satisfying $\Dm\subseteq\Wm\subseteq\Tm^\perp$) and optimal solutions of Problem~\ref{problem:DDPOF_N}. The result is reported in the following.
	
	Denote by $ \mathcal{P} = \bigcup_i \mathcal{P}_i $ the union of all $ \Dm$-to-$ \Tm$ paths. 
	
	\begin{theorem}[\textbf{\cite[Thm.~7]{lebon2026geometric}}]\label{thm:OF_otp_charact}
		Under Assumption~\ref{ass:2},
		\benu[label=(\alph*)]
		\item If $\{\Bm, \Cm\}$ is a feasible solution of Problem~\ref{problem:DDPOF_N} then \\$\Bm \supseteq \partial_+(\Wm,A)$ and $\Cm \supseteq \partial_-(\Wm,A)$;\label{item:a_OF_otp_charact}
		\item If $\{\Bm, \Cm\}$ is an optimal solution of Problem~\ref{problem:DDPOF_N} then \\$\Bm = \partial_+(\Wm,A)\cap \mathcal{P}=\partial_+(\Wm,A)$ and $\Cm = \partial_-(\Wm,A)\cap \mathcal{P}=\partial_-(\Wm,A)$;\label{item:b_OF_otp_charact}
		\item $\{\Bm, \Cm\}$ is an optimal solution of Problem~\ref{problem:DDPOF_N} iff $\Bm$ and $\Cm$ are the minimal cardinality sets of nodes for which $\Bm = \partial_+(\Wm,A)\cap \mathcal{P}=\partial_+(\Wm,A)$ and $\Cm = \partial_-(\Wm,A)\cap \mathcal{P}=\partial_-(\Wm,A)$.\label{item:c_OF_otp_charact}
		\eenu
	\end{theorem}
	
	As for the synthesis of the output feedback, we have the following result.
	\begin{theorem}[\textbf{\cite[Thm.~11]{lebon2026geometric}}]\label{thm:G}
		Under Assumption~\ref{ass:2}, and given a $(A,B)$- and $(C,A)$-invariant set of nodes $\Wm$, if $\Cm$ and $\Bm$ are two sets of output and input nodes, respectively, s.t. $\Cm=\partial_-(\Wm,A)$, and $\Bm=\partial_+(\Wm,A)$, then the unique friend of $\Wm$ solving the DDPOF, i.e., the matrix $G\in\mathbb{R}^{m\times p}$ s.t. $(A-BGC)\Wm\subseteq\Wm$, is given by $G=B^\top AC^\top$.
	\end{theorem}

	\section{Solving the minimal input/output cardinality DDPOF}
	What is missing in \cite{lebon2026geometric} is an efficient algorithm to compute the optimal $\{\Bm,\Cm\}$ that solves Problem~\ref{problem:DDPOF_N}. This algorithm, which is the main contribution of this paper, is described in this section. %To solve our minimal input/output cardinality DDPOF, we exploit Theorems~\ref{thm:OF_otp_charact} and \ref{thm:DDPN}.
	Before introducing the proposed algorithm, we first give the following Lemma, direct consequence of Theorems~\ref{thm:OF_otp_charact} and \ref{thm:DDPN}.
	\begin{lemma}\label{lem:int}
		Under Assumption~\ref{ass:2}, if $\{\Bm,\Cm\}$ is an optimal solution of Problem~\ref{problem:DDPOF_N}, then $\Bm\cap\Cm=\emptyset$, thus $|\Bm|+|\Cm|=|\Bm\,\cup\,\Cm|$.
	\end{lemma}
	\begin{proof}
		By assumption, $\{ \Bm, \Cm\} $ is a minimal cardinality solution of Problem~\ref{problem:DDPOF_N}. %, and hence it satisfies item~\ref{item:c_thm:ddpof} of Theorem~\ref{thm:DDPN}, and also item~\ref{item:c_OF_otp_charact} of Theorem~\ref{thm:OF_otp_charact}.
		By item~\ref{item:c_OF_otp_charact} of Theorem~\ref{thm:OF_otp_charact}, any minimal cardinality solution satisfies $\Bm = \partial_+(\Wm,A)$ and $\Cm = \partial_-(\Wm,A)$. By Definition~\ref{def:boundary}, $\partial_+(\Wm,A)\subseteq\Wm^\perp$ and $\partial_-(\Wm,A)\subseteq\Wm$. Since $\Wm$ and $\Wm^\perp$ are disjoint node sets, $\Bm\cap\Cm$ must be empty as well.
		%By contradiction, assume that there exists in $\Gm$ a $\Dm$-to-$\Tm$ path $\Pm_j=\{v_1,\dots,v_{\kappa_j}\}$ of length $\kappa_j-1\geq1$, with $v_1\in\Dm$ and $v_{\kappa_j}\in\Tm$, a sub-path of the form $\{v_{p-1},v_{p},v_{p+1}\}$ with $v_{p-1}\in\Cm$, $v_{p+1}\in\Bm$, and $v_{p}$ being both in $\Cm$ and in $\Bm$, $p\in\{1,\dots,\kappa_j-1\}$, $j\in\{1,\dots,\ell\}$. However, this means that for the path $\Pm_j$ both the sub-path $\{v_{p-1},v_{p}\}$ with $v_{p-1}\in\Cm$ and $v_{p}\in\Bm$, and the sub-path $\{v_{p},v_{p+1}\}$ with $v_{p}\in\Cm$ and $v_{p}\in\Bm$ satisfy the feasibility condition~\ref{item:c_thm:ddpof} of Theorem~\ref{thm:DDPN}. The existence of multiple sub-paths of $\Pm_j$ of length 1 is a contradiction of minimality.
	\end{proof}
	
	\subsection{Min-cut/max-flow algorithm revisited}\label{sec:mincut}
	
	What Theorems~\ref{thm:OF_otp_charact} and \ref{thm:DDPN} indicate is that, for each $\Dm$-to-$\Tm$ path $\Pm_j$, at least a node of $\Bm$ and a node $\Cm$ must belong to $\Pm_j$ and be pair-wise adjacent in $\Pm_j$, with the node of $\Cm$ coming before the node of $\Bm$ in $\Pm_j$. The solution pair $\{ \Bm, \Cm\} $ thus corresponds to a node cut separating $ \Dm $ from $ \Tm$. What Problem~\ref{problem:DDPOF_N} requires is that such node cut must have minimal cardinality. The problem can be seen as a special case of a min-cut/max-flow problem, whose optimal solution can be computed in polynomial time (see Table 1 of \cite{goldberg1988new}). 
	The strategy to achieve this decoupling with the minimal number of input and output nodes is summarized in Algorithms~\ref{alg:one}--\ref{alg:three}.

	\begin{algorithm}[hbt!]
		\caption{DDPOF by min-cut/max-flow}\label{alg:one}
		\begin{algorithmic}
			\State \textbf{Data:} {$\mathcal{G}$, $\mathcal{D}$, $\mathcal{T}$}
			\State \textbf{Result:} {$\Bm,\Cm$}
			\State \underline{$[\overline{\mathcal{G}};v^\Dm; v^\Tm]\gets$  \textbf{Algorithm~\ref{alg:two}}};\hspace{8pt} \textit{// Given a network $\mathcal{G}$ with $n$ nodes and $q$ edges and the nodes subsets $\mathcal{D}$ and $\mathcal{T}$, construct an extended network $\overline{\mathcal{G}}$}
			%with $2n+2$ nodes and $q+n+|\mathcal{D}|+|\mathcal{T}|$ edges with appropriate weights. The extra nodes are obtained by duplicating the existing nodes and adding a virtual source node $ v^\Dm$ and a virtual sink node $ v^\Tm$ connected to all disturbance, resp. target nodes.}
		\vspace{0.1cm}
		\State \underline{$[\overline{\Jm}^\Dm;\overline{\Jm}^\Tm]\gets$ \texttt{maxflow}($\overline{\mathcal{G}},v^\Dm,v^\Tm$)};\hspace{8pt} \textit{// Using the max-flow/min-cut algorithm on $\overline{\mathcal{G}}$, partition its nodes into those reachable and not reachable from $v^\Dm$ in the residual graph induced by the computed maximum flow, resp. $\overline{\Jm}^\Dm$ and $\overline{\Jm}^\Tm$}
		\State \underline{$[\Bm;\Cm]\gets$ \textbf{Algorithm~\ref{alg:three}}}.\hspace{8pt} \textit{// Given $\overline{\Jm}^\Dm$ and $\overline{\Jm}^\Tm$, find the minimal cardinality set of control input nodes and output nodes $\Bm,\Cm$ that solves the DDPOF}
	\end{algorithmic}
\end{algorithm}

\begin{algorithm}[hbt!]
	\caption{Construct the extended network $\overline{\mathcal{G}}$ }\label{alg:two}
	\begin{algorithmic}
		\State \textbf{Data:} {$\mathcal{G}$, $\mathcal{D}$, $\mathcal{T}$}
		\State \textbf{Result:} {$\overline{\mathcal{G}}= ( \overline{\Vm}, \overline{\Em}, \overline{A}) $, $v^\Dm$, $v^\Tm$}
		\State Construct $ \overline{\Vm}$: $\overline{\Vm}=\Vm\cup\{v_{n+1}, \, \ldots, v_{2n}\}\cup\{v_{2n+1}, \, \ldots, v_{3n}\}\cup\{v_{3n+1},v_{3n+2}\}$;%\Comment{Triple the network nodes and add extra sink and source nodes}
		\vspace{-0mm}
		\State Construct $ \overline{\Em}$: 
		\bite
		\item If $\exists\,(v_i,\,v_j)\in\Em$ then $\exists\,(v_{i},\,v_{j+2n}),\,(v_{i+n},\,v_{j}) \in\overline{\Em}$;%\Comment{Include original edges}
		\vspace{-0mm}
		\item $(v_{i+2n},\,v_{i}),\,(v_i,\,v_{i+n})\in\overline{\Em},\,\forall\,i\in\{1,\dots,n\}$;%\Comment{Connect each pair of triplicated nodes}
		\vspace{-0mm}
		\item $(v_{3n+1},\,v_{d})\in\overline{\Em},\,\forall\,v_d\in\Dm$;%\Comment{Connect the source node to all nodes with disturbances}
		\vspace{-0mm}
		\item $(v_t,\,v_{3n+2})\in\overline{\Em},\,\forall\,v_t\in\Tm$;%\Comment{Connect all target nodes to the sink node}
		\vspace{-0mm}
		\eite
		\State Construct $ \overline{A}$:
		\vspace{-0mm}
		\bite
		\item If $\exists\,(v_i,\,v_j)\in\Em$ then $ \overline{A}_{v_j,v_{i+n}} =\overline{A}_{v_{j+2n},v_i}=+\infty$;%\Comment{\,\,Put all ``true'' edge weights to infinity}
		\vspace{-0mm}
		\item $ \forall \; v_i \in \Vm $, $ \overline{A}_{v_{i+n}, v_i} =\overline{A}_{v_{i}, v_{i+2n}}=1 $;%\Comment{Put weight 1 in all edges between triples}
		\vspace{-0mm}
		\item $ \forall v_i \in \Dm $, $\overline{A}_{v_i,v_{3n+1}}=+\infty$;%\Comment{Put infinite weights in all new edges between source node and nodes in $ \Dm$}
		\vspace{-0mm}
		\item $ \forall v_i \in \Tm $, $\overline{A}_{v_{3n+2},v_i}=+\infty$;%\Comment{Put infinite weights in all new edges between nodes in $ \Tm $ and sink node}
		\vspace{-0mm}
		
		\eite
		\State $ v^\Dm = v_{3n+1}$; $ v^\Tm = v_{3n+2}$.
		\vspace{0mm}
	\end{algorithmic}
\end{algorithm}

\begin{algorithm}[hbt!]
	\caption{Minimal input/output solution $\{\Bm,\Cm\}$}\label{alg:three}
	\begin{algorithmic}
		\State \textbf{Data:} {$\overline{\Jm}^\Dm$, $\overline{\Jm}^\Tm$, $n$}
		\State \textbf{Result:} {$\Bm, \Cm$}
		\State $\Bm=\emptyset$; $\Cm=\emptyset$;\hspace{8pt} \textit{// Store into $\Bm$ resp. $\Cm$ the head resp. the tail nodes of the edges between $\overline{\Jm}^\Dm$ and $\overline{\Jm}^\Tm$}
		\For{$i=1:n$}
		\If{$v_i\in\overline{\Jm}^\Dm$ and $v_{i+n}\in\overline{\Jm}^\Tm$}
		\State{$\Cm=\Cm \cup \{v_i\}$;}\hspace{8pt} \textit{// Keep tail of edge $(v_i,v_{i+n})$}
		\EndIf
		\If{$v_{i}\in\overline{\Jm}^\Tm$ and $v_{i+2n}\in\overline{\Jm}^\Dm$}
		\State{$\Bm=\Bm \cup \{v_i\}$;}\hspace{8pt} \textit{// Keep head of edge $(v_{i+2n},v_i)$}
		\EndIf\EndFor
		.
	\end{algorithmic}
\end{algorithm}

%Under Assumption~\ref{ass:2}, the algorithms always converge in finite time to an optimal solution, since the number of nodes and edges is finite.

Under Assumption~\ref{ass:2} and given a network $\Gm=(\Vm,\Em,A)$ with $n=|\Vm|$ nodes and $q=|\Em|$ edges, the extended network $\overline{\mathcal G}$ has $3n+2$ nodes and
$2(n+q)+|\mathcal D|+|\mathcal T|=O(n+q)$ edges, so Algorithm~\ref{alg:two} runs in $O(n+q)$ time.
The minimum cut value is at most $n$, since $\mathcal W=\mathcal D$ is feasible with
$\mathcal C=\partial_-(\mathcal D,A)\subseteq\mathcal D$ and
$\mathcal B=\partial_+(\mathcal D,A)\subseteq\mathcal V\setminus\mathcal D$ by Lemma~\ref{lem:feas_to_cut}. Hence the weight $+\infty$ can be replaced by $n+1$ without changing the minimum cut, and a min-cut/max-flow algorithm based on augmenting paths~\cite{ford1956maximal} therefore runs in $O(n(n+q))$ time, and any algorithm in Table~1 of~\cite{goldberg1988new} can be used instead. Algorithm~\ref{alg:three} scans $i=1,\dots,n$ once, with $O(1)$ membership tests on $\overline{\Jm}^\Dm$
(a Boolean vector filled during the residual search), so it runs in $O(n)$.
The overall cost is that of one min-cut/max-flow computation on $\overline{\mathcal G}$.

Here are some more details on some parts of the algorithms.

\noindent\textbf{Algorithm~\ref{alg:two}: Construct the extended network $\bm{\overline{\mathcal{G}}}$.}
\label{sec:extended} The minimal input/output node cut problem can be transformed into a minimum edge cut problem on an extended network $\overline{\Gm}=(\overline{\Vm},\overline{\Em},\overline{A})$, where it can be solved via a standard min-cut/max-flow algorithm. 
% Since we are interested in finding a minimal node-cut (whereas a min-cut/max-flow algorithm produces a minimal edge cut), the network $\Gm$ must be expanded suitably into a new extended network $\overline{\Gm}=(\overline{\Vm},\overline{\Em},\overline{A})$ allowing to identify such nodes by means of a minimal edge-cut.
To do so, for each node $v_i$, two new nodes $v_{i+n}$ and $v_{i+2n}$, and new edges $(v_i,\,v_{i+n})$ and $(v_{i+2n},\,v_{i})$ are generated, with unit weights. For all preexisting edges in $\Vm$ between nodes $v_i$ and $v_j$, an edge is generated in $\overline{\Gm}$ between $v_i$ and $v_{j+2n}$, and an edge is generated in $\overline{\Gm}$ between $v_{i+n}$ and $v_{j}$. For those two edges the weight is put at infinity (practically, at a very high number), see Fig.~\ref{fig:mini}. In such a way, an edge cut involving $(v_i,\,v_{i+n})$ and $(v_{j+2n},\,v_{j})$ corresponds to selecting the node $v_i\in\Cm$ and $v_j\in\Bm$. Two additional nodes must be added, $v^\Dm=v_{3n+1}$ and $v^\Tm=v_{3n+2}$, so that $|\overline{\Vm}|=3n+2$. Node $v^\Dm$ acts as a source node and is connected via $|\Dm|$ edges to all disturbance nodes in $\Dm$. All target nodes in $\Tm$ are connected via $|\Tm|$ edges to the sink node $v^\Tm$. The total number of edges in $\overline{\Em}$ is $2(n+q)+|\Dm|+|\Tm|$. On $ \overline{\Gm} $ all paths from $ v^\Dm $ to $ v^\Tm $ have even length, i.e., must contain an odd number of nodes.

\noindent\textbf{Min-cut/Max-flow.} The min-cut/max-flow algorithm is a minimum edge cut algorithm performed on $ \overline{\Gm}$ from the source node $ v^\Dm $ to the sink node $ v^\Tm $. %The cut partitions the network nodes in two disjoint sets, one, $ \overline{\Jm}^\Dm$, containing the nodes that are reachable from the disturbances, the other, $ \overline{\Jm}^\Tm $, containing the targets and all the nodes that cannot be reached by the disturbances.
Let $(\overline{\Jm}^\Dm,\overline{\Jm}^\Tm)$ be a minimum $v^{\mathcal D}$-$v^{\mathcal T}$ cut of $\overline{\mathcal G}$. Any min-cut/max-flow algorithm returns one such cut: given a maximum flow, $\overline{\Jm}^\Dm$ is the set of nodes reachable from $v^{\mathcal D}$ in the residual graph~\cite{ford1956maximal}, and $\overline{\Jm}^\Tm=\overline{\mathcal V}\setminus \overline{\Jm}^\Dm$.%A standard min-cut/max-flow algorithm (e.g., \cite{ford1956maximal}) is performed on $\overline{\Gm}=(\overline{\Vm},\overline{\Em}, \overline{A})$, from the source node $ v^\Dm $ to the sink node $ v^\Tm $. 

\noindent\textbf{Algorithm~\ref{alg:three}: find the minimal input/output $\{\Bm,\Cm\}$.} The meaning of Algorithm~\ref{alg:three} is explained in the following proposition.
\begin{proposition}
	\label{prop:cut-extended}
	Given $ \Gm $ and its extended network $ \overline{\Gm}$, an edge cut in $ \overline{\Gm}$ corresponds to a node cut in $ \Gm$. The min-cut value on $ \overline{\Gm} $ is equal to $ |\Bm|+|\Cm|$, i.e., to the sum of the cardinalities of the sets of input and output nodes associated to the solution.
\end{proposition}

\begin{proof}
	In $ \overline{\Gm}$ all ``true'' edges of $ \Gm $ are duplicated and have $ \infty $ weight, hence they are never selected by the min-cut algorithm. Therefore all edges on which each edge cut is computed in $ \overline{\Gm}$ must have the form $(v_i,v_{i+n})$ or $(v_{i+2n},v_i)$, and these edge weights are all equal to 1. By construction, the value of each cut is equal to the number of edges across the cut and equal also to the number of nodes in the associated node cut in $ \Gm$. This implies that the min-cut value is equal to $ |\Bm|+|\Cm|$. 
\end{proof}

\begin{figure}[htb!]
	\centering
	\includegraphics[clip=true, width=0.6\linewidth]{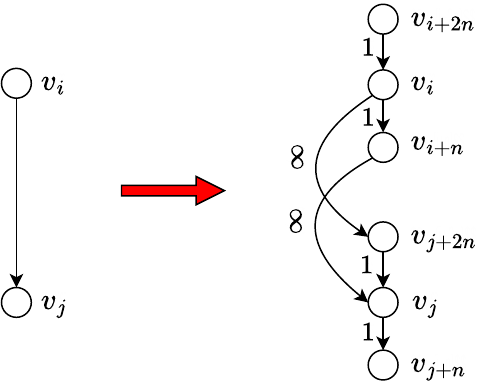}
	\caption{Algorithm~\ref{alg:two}. Representation of an edge $(v_i,v_j)\in\Em$ in the extended network $\overline{\Gm}$.}
	\label{fig:mini}
	%\vspace{-0.35cm}
\end{figure}

\subsection{Optimality}
We start by proving that a feasible solution $\{\Bm,\Cm\}$ of Problem~\ref{problem:DDPOF_N} identifies a $v^{\mathcal D}$-$v^{\mathcal T}$ cut with value at most $|\mathcal B|+|\mathcal C|$ in $\overline{\mathcal G}$.
\begin{lemma}\label{lem:feas_to_cut}
	Let $\{\mathcal B,\mathcal C\}$ be a feasible solution of Problem~\ref{problem:DDPOF_N}.
	Then $\overline{\mathcal G}$ admits a $v^{\mathcal D}$-$v^{\mathcal T}$ cut with value
	at most $|\mathcal B|+|\mathcal C|$.
\end{lemma}

\begin{proof}
	The edges with unit weights are $(v_i,v_{i+n})$ and $(v_{i+2n},v_i)$, $i\in\{1,\dots,n\}$. All remaining edges have infinity weights and belong to one of four
	families:
	\begin{enumerate}
		\item[(i)] $(v^{\mathcal D},v_d)$, $v_d\in\mathcal D$;
		\item[(ii)] $(v_t,v^{\mathcal T})$, $v_t\in\mathcal T$;
		\item[(iii)] $(v_i,v_{j+2n})$ for $(v_i,v_j)\in\mathcal E$;
		\item[(iv)] $(v_{i+n},v_j)$ for $(v_i,v_j)\in\mathcal E$.
	\end{enumerate}
	Let $\Wm$ be the $(A,B)$- and $(C,A)$-invariant node set of Problem~\ref{problem:DDPOF_N} s.t. $\Dm\subseteq\Wm\subseteq\Vm\setminus\Tm$.
	From item~\ref{item:a_OF_otp_charact} of Theorem~\ref{thm:OF_otp_charact}, feasibility implies $\partial_-(\mathcal W,A)\subseteq\mathcal C$ and
	$\partial_+(\mathcal W,A)\subseteq\mathcal B$. Let $(v_i,v_j)\in\mathcal E$ with $v_i\in\mathcal W$ and $v_j\notin\mathcal W$.
	By item~\ref{item:a_prop:CA-invar-equiv-char} of Proposition~\ref{prop:CA-invar-equiv-char}, $v_i\notin\mathcal C$ would imply $v_j\in\mathcal W$; hence
	$v_i\in\mathcal C$. By item~\ref{item:a_prop:AB-invar-equiv-char} Proposition~\ref{prop:AB-invar-equiv-char}, $v_i\in\mathcal W$ implies
	$v_j\in\mathcal W\cup\mathcal B$; hence $v_j\in\mathcal B$. Every element of
	$\partial_-(\mathcal W,A)$ (resp. $\partial_+(\mathcal W,A)$) is the tail (resp. head)
	of such an edge.
	
	With $N^+(\mathcal W):=\{v_j\;\text{s.t.}\;\exists\, v_i\in\mathcal W,\ (v_i,v_j)\in\mathcal E\}$, define the set $\overline{\Jm}^\Dm := \{v^{\mathcal D}\}\,\cup\,\Wm\,\cup\,\{v_{i+n}\in\overline{\Vm}\;\text{s.t.}\;v_i\in\mathcal W\setminus\partial_-(\mathcal W,A)\}\,\cup\,\{v_{j+2n}\in\overline{\Vm}\;\text{s.t.}\;v_j\in N^+(\mathcal W)\}$.
	Then $v^{\mathcal D}\in \overline{\Jm}^\Dm$, and $v^{\mathcal T}\notin \overline{\Jm}^\Dm$ by construction, so $(\overline{\Jm}^\Dm,\overline{\Jm}^\Tm)$, with $\overline{\Jm}^\Tm=\overline{\mathcal V}\setminus\overline{\Jm}^\Dm$, defines
	a $v^{\mathcal D}$-$v^{\mathcal T}$ cut in $\overline{\mathcal G}$. Let $\delta(\overline{\Jm}^\Dm,\overline{\Jm}^\Tm)$ be the edge cut of $(\overline{\Jm}^\Dm,\overline{\Jm}^\Tm)$, i.e., the set of edges with tail in
	$\overline{\Jm}^\Dm$ and head in $\overline{\mathcal V}\setminus \overline{\Jm}^\Dm$.
	
	No edge with infinite weight is in $\delta(\overline{\Jm}^\Dm,\overline{\Jm}^\Tm)$:
	\begin{enumerate}
		\item[(i)] $v_d\in\mathcal D\subseteq\mathcal W$, so $v_d\in \overline{\Jm}^\Dm$.
		\item[(ii)] The tail $v_t$ has $v_t\in\mathcal T$, and $\mathcal W\subseteq \Tm^\perp$,
		so $v_t\notin \overline{\Jm}^\Dm$.
		\item[(iii)] If $v_i\in \overline{\Jm}^\Dm$ then $v_i\in\mathcal W$, hence $v_j\in N^+(\mathcal W)$ and $v_{j+2n}\in \overline{\Jm}^\Dm$.
		\item[(iv)] If $v_{i+n}\in \overline{\Jm}^\Dm$ then $v_i\in\mathcal W\setminus\partial_-(\mathcal W,A)$, i.e., $v_i$ has no out-neighbour in $\mathcal W^\perp$. Hence $v_j\in\mathcal W$, so $v_j\in \overline{\Jm}^\Dm$.
	\end{enumerate}
	
	The edges with unit weights in $\delta(\overline{\Jm}^\Dm,\overline{\Jm}^\Tm)$:
	\begin{itemize}
		\item $(v_i,v_{i+n})$ leaves $\overline{\Jm}^\Dm$ iff $v_i\in \overline{\Jm}^\Dm$ and $v_{i+n}\notin \overline{\Jm}^\Dm$.
		Now $v_i\in \overline{\Jm}^\Dm\iff v_i\in\mathcal W$, and for such $i$,
		$v_{i+n}\notin \overline{\Jm}^\Dm\iff v_i\in\partial_-(\mathcal W,A)$.
		These edges are in bijection with $\partial_-(\mathcal W,A)$.
		\item $(v_{j+2n},v_j)$ leaves $\overline{\Jm}^\Dm$ iff $v_{j+2n}\in \overline{\Jm}^\Dm$ and $v_j\notin \overline{\Jm}^\Dm$,
		i.e., $v_j\in N^+(\mathcal W)\setminus\mathcal W=\partial_+(\mathcal W,A)$.
		These edges are in bijection with $\partial_+(\mathcal W,A)$.
	\end{itemize}
	It follows that the value of the $v^{\mathcal D}$-$v^{\mathcal T}$ cut $(\overline{\Jm}^\Dm,\overline{\Jm}^\Tm)$ is $|\partial_-(\mathcal W,A)|+|\partial_+(\mathcal W,A)|
	\;\le\; |\mathcal C|+|\mathcal B|$.
\end{proof}

%\rednote{The following corollary proves that an optimal solution $\{\Bm,\Cm\}$ of Problem~\ref{problem:DDPOF_N} identifies a $v^{\mathcal D}$-$v^{\mathcal T}$ cut with minimum value equal to $|\mathcal B|+|\mathcal C|$ in $\overline{\mathcal G}$.}
%\begin{corollary}\label{cor:rev_opt}
%\rednote{Let $\{\mathcal B,\mathcal C\}$ be an optimal solution of Problem~\ref{problem:DDPOF_N}.
	%Then $\overline{\mathcal G}$ admits a $v^{\mathcal D}$-$v^{\mathcal T}$ cut with minimum value
	%of $|\mathcal B|+|\mathcal C|$.}
%\end{corollary}
%\begin{proof}
%\rednote{The proof follows as a consequence of Lemma~\ref{lem:feas_to_cut} and item~\ref{item:b_OF_otp_charact} of Theorem~\ref{thm:OF_otp_charact}.} 
%\end{proof}

The optimality of the solution is summarized in the following statement.

\begin{theorem}\label{thm:main}
	Under Assumption~\ref{ass:2}, the solution $\{\Bm,\Cm\}$ computed by means of Algorithms~\ref{alg:one}--\ref{alg:three} is an optimal solution of Problem~\ref{problem:DDPOF_N}. 
\end{theorem}
\begin{proof}
	We start by noting that an edge cut in $\overline{\mathcal G}$ can only involve edges of the
	form $(v_i,v_{i+n})$ or $(v_{i+2n},v_i)$, $i\in\{1,\dots,n\}$, since all other edges are at
	$+\infty$ (a finite cut exists by Lemma~\ref{lem:feas_to_cut}, since $\mathcal W=\mathcal D$ is feasible by Proposition~\ref{prop:cont_cond}). 
	From Proposition~\ref{prop:cut-extended}, an edge cut involving $(v_i,v_{i+n})$ in $\overline{\mathcal G}$ identifies a node of $\mathcal C$ in $\mathcal G$, while $(v_{i+2n},v_i)$ in
	$\overline{\mathcal G}$ identifies a node of $\mathcal B$ in $\mathcal G$. Alternatively, saying
	that $\{\mathcal B,\mathcal C\}$ identifies a node cut for $\mathcal G$ means that the set of edges $\mathcal H=\{(v_i,v_{i+n})\;\text{s.t.}\;v_i\in\mathcal C\}\cup\{(v_{j+2n},v_j)\;\text{s.t.}\;v_j\in\mathcal B\}$
	is an edge cut for $\overline{\mathcal G}$, i.e., when removed, there are no more directed paths
	connecting $v^{\mathcal D}$ to $v^{\mathcal T}$. If we call
	$\hat{\mathcal B}=\{v_{i+2n}\;\text{s.t.}\;v_i\in\mathcal B\ \text{and}\ (v_{i+2n},v_i)\in\mathcal H\}$, and $\hat{\mathcal C}=\{v_{i+n}\;\text{s.t.}\;v_i\in\mathcal C\ \text{and}\ (v_i,v_{i+n})\in\mathcal H\}$
	then every edge of $\mathcal H$ has its tail in $\mathcal C\cup\hat{\mathcal B}$ and its head in
	$\mathcal B\cup\hat{\mathcal C}$. Let us assume that
	$\overline{\mathcal J}^{\mathcal D}\subseteq\overline{\mathcal V}$ and
	$\overline{\mathcal J}^{\mathcal T}=\overline{\mathcal V}\setminus\overline{\mathcal J}^{\mathcal D}$
	is the partition of $\overline{\mathcal V}$ inducing the edge cut, i.e.,
	$\delta(\overline{\mathcal J}^{\mathcal D},\overline{\mathcal J}^{\mathcal T})=\mathcal H$.
	Denote by $\overline{\mathcal E}^{\mathcal D}$ and $\overline{\mathcal E}^{\mathcal T}$ the sets
	of edges within $\overline{\mathcal J}^{\mathcal D}$ and $\overline{\mathcal J}^{\mathcal T}$,
	respectively. The set of edges with tail in $\overline{\mathcal J}^{\mathcal D}$ and head in
	$\overline{\mathcal J}^{\mathcal T}$ is denoted by $\overline{\mathcal E}^{\mathcal D\to\mathcal T}$,
	while the set of edges with tail in $\overline{\mathcal J}^{\mathcal T}$ and head in
	$\overline{\mathcal J}^{\mathcal D}$ is denoted by $\overline{\mathcal E}^{\mathcal T\to\mathcal D}$.
	By construction, it follows $\overline{\mathcal E}^{\mathcal D\to\mathcal T}=\mathcal H$. As a
	consequence, $\mathcal C,\hat{\mathcal B}\subseteq\overline{\mathcal J}^{\mathcal D}$ and
	$\mathcal B,\hat{\mathcal C}\subseteq\overline{\mathcal J}^{\mathcal T}$. In fact, if by
	contradiction a node in $\overline{\mathcal J}^{\mathcal T}$ is reached by an edge with tail in
	$\overline{\mathcal J}^{\mathcal D}\setminus(\mathcal C\cup\hat{\mathcal B})$, then it is also
	reachable from $v^{\mathcal D}$, which is a contradiction according to the definition of
	$\overline{\mathcal J}^{\mathcal T}$ given above. Similarly, if by contradiction a node in
	$\overline{\mathcal J}^{\mathcal T}\setminus(\hat{\mathcal C}\cup\mathcal B)$ is reached by an
	edge with tail in $\overline{\mathcal J}^{\mathcal D}$, then it is also reachable from
	$v^{\mathcal D}$, which is a contradiction according to the definition of
	$\overline{\mathcal J}^{\mathcal D}$ given above. By construction we write the extended network as
	$\overline{\mathcal G}=(\overline{\mathcal V},\overline{\mathcal E},\overline{A})
	=(\overline{\mathcal J}^{\mathcal D}\cup\overline{\mathcal J}^{\mathcal T},
	\overline{\mathcal E}^{\mathcal D}\cup\overline{\mathcal E}^{\mathcal T}\cup
	\overline{\mathcal E}^{\mathcal T\to\mathcal D}\cup\mathcal H,\overline{A})$.
	
	To prove that the minimum edge cut in $\overline{\mathcal G}$ corresponds to the optimal
	solution $\{\mathcal B,\mathcal C\}$ of Problem~\ref{problem:DDPOF_N}, we first show that this edge cut always
	corresponds to a feasible solution, and secondly that it identifies one with minimal
	cardinality. The converse implication is given in Lemma~\ref{lem:feas_to_cut}. By construction, if $\exists (v_i,v_j)\in\mathcal E$ then $\exists (v_i,v_{j+2n}),(v_{i+n},v_j)\in\overline{\mathcal E}$. 
	We claim that a minimum edge cut from $v^{\mathcal D}$ to $v^{\mathcal T}$ in $\overline{\mathcal G}$ must identify nodes in $\mathcal G$ that satisfy item~\ref{item:c_thm:ddpof} of Theorem~\ref{thm:DDPN}, thus proving the existence of a set $\mathcal W$ both controlled and conditioned invariant in $\mathcal V$ that contains $\mathcal D$ and is contained in $\mathcal V\setminus\mathcal T$. Assume by contradiction that there is a $\mathcal D$-to-$\mathcal T$ path
	$\mathcal P_j=\{v_1,\dots,v_{\kappa_j}\}$ of length $\kappa_j-1\ge1$, with $v_1\in\mathcal D$ and
	$v_{\kappa_j}\in\mathcal T$, such that no sub-path of length 1 of the form $\{v_p,v_{p+1}\}$ has
	$v_p\in\mathcal C$ and $v_{p+1}\in\mathcal B$, $p\in\{1,\dots,\kappa_j-1\}$. For each such
	$p$, $\overline{\mathcal G}$ contains two routes from $v_p$ to $v_{p+1}$, each made of one
	unit-weight edge and one edge of weight $+\infty$: $v_p\to v_{p+n}\to v_{p+1}\;(\text{unit edge }(v_p,v_{p+n}))$, $v_p\to v_{p+1+2n}\to v_{p+1}\; (\text{unit edge }(v_{p+1+2n},v_{p+1}))$.
	The first route contains an edge of $\mathcal H$ iff $v_p\in\mathcal C$, the second iff $v_{p+1}\in\mathcal B$. By assumption, it is not true that both $v_p\in\mathcal C$ and
	$v_{p+1}\in\mathcal B$; hence $v_p\notin\mathcal C$ or $v_{p+1}\notin\mathcal B$, i.e., at
	least one of the two routes avoids $\mathcal H$. Concatenating these routes for
	$p=1,\dots,\kappa_j-1$ with the edges $(v^{\mathcal D},v_1)$ and $(v_{\kappa_j},v^{\mathcal T})$, both of weight $+\infty$ and hence not in $\mathcal H$, gives a
	path from $v^{\mathcal D}$ to $v^{\mathcal T}$ in $\overline{\mathcal G}$ that does not use any
	edge of $\mathcal H$. Then $\mathcal H$ is not an edge cut, which is a contradiction. Hence
	every $\mathcal D$-to-$\mathcal T$ path satisfies item~\ref{item:c_thm:ddpof} of Theorem~\ref{thm:DDPN}, i.e., $\{\mathcal B,\mathcal C\}$ is feasible. Finally, optimality arises from the minimality of the edge cut's cardinality. By contradiction
	assume that there exists a feasible $\{\mathcal B',\mathcal C'\}$ with
	$|\mathcal B'|+|\mathcal C'|<|\mathcal B|+|\mathcal C|=|\mathcal H|$. By Lemma~\ref{lem:feas_to_cut},
	$\overline{\mathcal G}$ admits a $v^{\mathcal D}$-$v^{\mathcal T}$ cut with value at most
	$|\mathcal B'|+|\mathcal C'|<|\mathcal H|$, which contradicts the minimality of $\mathcal H$.
	Hence $\{\mathcal B,\mathcal C\}$ is an optimal solution of Problem~\ref{problem:DDPOF_N}.
\end{proof}

\begin{remark}
	It is important to highlight that, in general, the optimal solution $\{\Bm,\Cm\}$ of Problem~\ref{problem:DDPOF_N} is not unique. In order to obtain all optimal solutions, specialized techniques and extensions of the min-cut/max-flow algorithm (e.g., residual graph analysis, Gomory-Hu trees \cite{gomory1961multi}) can be used to enumerate all equal-cardinality cuts \cite{abboud2023all,Picard1980}. To discriminate among such optima, a secondary objective function must be introduced. 
\end{remark}

\begin{remark}
	An optimal solution $\{\Bm,\Cm\}$ of Problem~\ref{problem:DDPOF_N} does not identify a minimum edge cut in $\Gm$, which can be instead computed via a standard min-cut from $\Dm$ to $\Tm$ in $\Gm$ after setting all the edge weights to 1 (see Example~\ref{ex:ex1}). For instance, if the optimal solution $\{\Bm,\Cm\}$ of Problem~\ref{problem:DDPOF_N} is not unique, it is possible to introduce a secondary objective function minimizing the number of edges to be cut.
\end{remark}

\section{Numerical Example}
In this numerical section we give an example of application of  Algorithms~\ref{alg:one}--\ref{alg:three} for computing the minimal input/output cardinality DDPOF.
\begin{example}
	\label{ex:ex1}
	Let us consider the network $ \Gm $ of $10$ nodes and $\Dm~=~\{v_1,v_2,v_3\}$ and $\Tm=\{v_8,v_9,v_{10}\}$ shown in Fig.~\ref{fig:ex1a} on the left.
	Using Algorithm~\ref{alg:two} the extended network $\overline{\Gm}$ shown in Fig.~\ref{fig:ex1a} on the right is constructed.
	\begin{figure}[htb!]
		\centering
		\includegraphics[clip=true, width=1\linewidth]{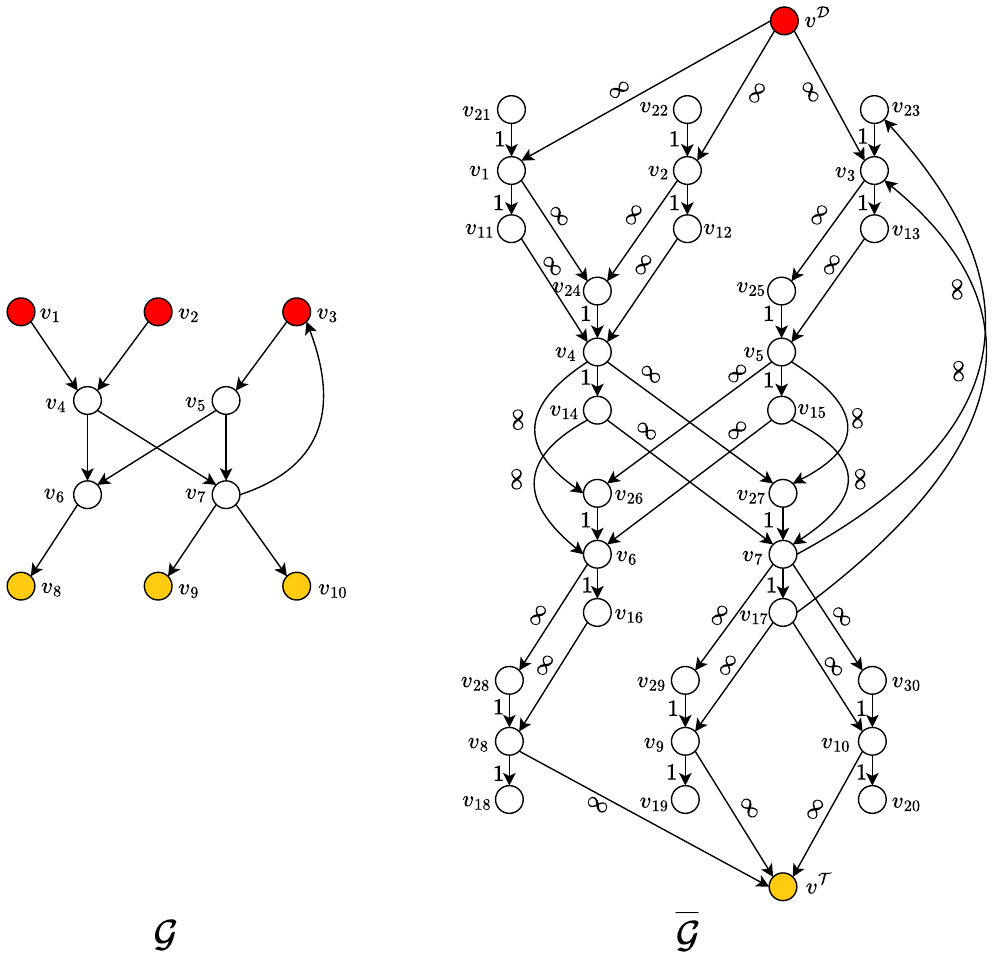}
		\caption{Example~\ref{ex:ex1}. Disturbances are in red and the targets are in yellow.}
		\label{fig:ex1a}
		%\vspace{-0.35cm}
	\end{figure}
	\begin{figure}[htb!]
		\centering
		\includegraphics[clip=true, width=1\linewidth]{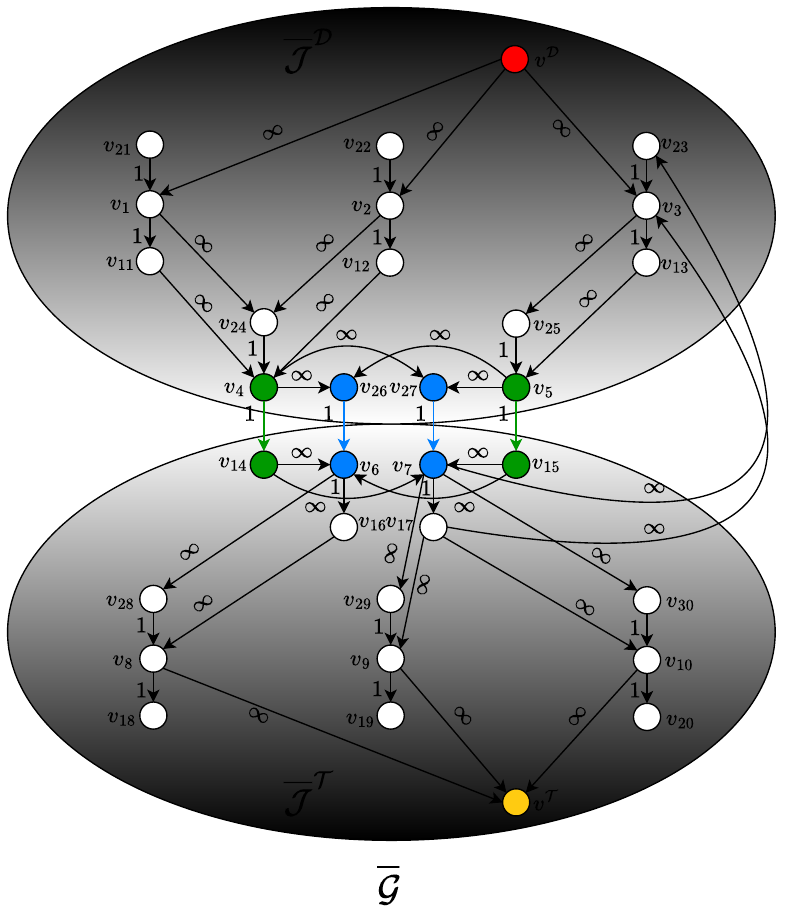}
		\caption{Example~\ref{ex:ex1}. Disturbances are in red, targets in yellow, whereas the control inputs are in blue and the output nodes are in green. Nodes $v_{21}$, $v_{22}$, $v_{23}$ are drawn in $\overline{\Jm}^\Dm$ for readability; they belong to $\overline{\Jm}^\Tm$ according to the definition, and both choices give the same edge cut.}
		\label{fig:ex1b}
		%\vspace{-0.35cm}
	\end{figure}\\
	The minimal cardinality solution of Problem~\ref{problem:DDPOF_N} is $\Bm=\{v_6,v_7\}$, $\Cm=\{v_4,v_5\}$, as depicted in Fig.~\ref{fig:ex1b}. In fact, the set of edges $\Hm=\{(v_{26},v_{6}),\,(v_{27},v_{7}),\,(v_4,v_{14}),\,(v_5,v_{15})\}$ is the minimum edge cut in the extended network $\overline{\Gm}$. Note that $\hat{\Bm}=\{v_{26},v_{27}\}\subseteq\overline{\Jm}^\Dm$, and $\hat{\Cm}=\{v_{14},v_{15}\}\subseteq\overline{\Jm}^\Tm$, as in the proof of Theorem~\ref{thm:main}. Moreover, note that the run of a classical edge min-cut/max-flow from $\Dm$ to $\Tm$ in $\Gm$ would have provided the feasible solutions $\Bm=\{v_4,v_5\}$, $\Cm=\Dm$, or $\Bm=\Tm$, $\Cm=\{v_6,v_7\}$, both suboptimal w.r.t. the cost function $| \Bm|+|\Cm|$.
	\hfill \qed
\end{example}
\section{Conclusions}
Protecting strategic nodes from disturbances is a key challenge in cybersecurity and networked control. This paper addressed the DDPOF over networks by developing an algorithm that determines, in polynomial time, the minimal cardinality sets of control input and measurement output nodes required to reject the disturbances. MATLAB functions implementing the results presented in this article are available in the Zenodo repository \cite{lebon_2026_23476722}.
\bibliographystyle{abbrv}
\bibliography{biblio}

\end{document}